\documentclass[a4paper,11pt]{article}
\usepackage{auxlib}

\usepackage{fullpage}
\usepackage{times}

\usepackage{url}
\usepackage{iftex}\ifPDFTeX\usepackage[utf8]{inputenc}\fi
\usepackage{color}
\usepackage{amsmath,amssymb,amsthm}

\usepackage{empheq}
\usepackage{float}       
\usepackage{algorithm}
\usepackage{algpseudocode}

\usepackage{booktabs,array,tabularx,longtable,placeins}
\usepackage{microtype}
\usepackage{enumitem}
\usepackage{algorithm}
\usepackage{algpseudocode}

\algtext*{EndIf}
\algtext*{EndWhile}
\algtext*{EndFor}

\usepackage{libertine}

\usepackage{iftex}
\ifPDFTeX

\usepackage[a4paper, total={170mm,257mm}, left=20mm, top=20mm]{geometry}

\newlength\mylen

\def\QQ{\mathbb Q}

\def\cH{\mathcal H}

\def\cM{\mathcal M}
\def\cN{\mathcal N}

\newcommand{\np}{\ensuremath{\mathrm{NP}}}

\newcommand{\LP}{\ensuremath{\mathsf{LP}}}

\newcommand{\WMMS}{\operatorname{WMMS}}

\newcommand{\NN}{\mathcal N}
\newcommand{\MM}{\mathcal M}

\newcommand{\astar}{\alpha_\star}

\newcommand{\dd}{\,\mathrm d}

\title{Improving the Bounds on Approximating Weighted Maximin Share for Chores}
\author{
	Trung Thanh Nguyen\thanks{Faculty of Data Science and Artificial Intelligence, National Economics University, Hanoi 11616, Vietnam} 
	\and 
	Khaled Elbassioni\thanks{Khalifa University of Science and Technology, Abu Dhabi 127788, UAE} 
	\and 
	J\"{o}rg Rothe\thanks{Institut f\"{u}r Informatik, MNF, Heinrich-Heine-Universit\"{a}t D\"{u}sseldorf, 40225, D\"{u}sseldorf, Germany} 
}
\date{\today}

\begin{document}
	\maketitle
	\thispagestyle{plain}
	
	\begin{abstract}
		We investigate the problem of allocating indivisible chores among agents with arbitrary positive weights and additive nonnegative cost functions. An allocation is said to be $\alpha$-weighted-maximin-share (WMMS) fair if the cost incurred by each agent is at most $\alpha$ times her weighted maximin share. We establish that every instance admits a $4.124234528$-WMMS allocation. Moreover, for instances in which all weights and costs are rational numbers represented in binary, we give an algorithm that computes a $4.12424$-WMMS allocation in time polynomial in the number of agents, the number of chores, and the total binary encoding length of the input. Complementing these upper bounds, we construct a family of instances with strictly positive, identically ordered costs for which the optimal WMMS approximation factors converge to $e$. Together, our results substantially narrow the previously best-known gap of $[2,12]$ established by Li et al.~\cite{LWX26} to $[e,4.12424]$.
	\end{abstract}
	
	
	\section{Introduction}\label{sec:introduction}
	Fair division is a fundamental problem in economics, mathematics, and computer science, concerned with allocating resources or responsibilities among agents in a manner that is perceived as equitable~\cite{lang-rothe-2015,maudet-2016,AmanatidisBFV22}. It arises naturally in a wide range of applications, including the allocation of public resources, the distribution of household duties, workload assignment, and the sharing of undesirable tasks~\cite{GoldmanP14}. A central challenge is to formalize fairness when agents have heterogeneous preferences or costs. Several influential fairness criteria have been proposed to address this challenge. {\em Envy-freeness} (EF) requires that no agent prefer another agent's allocation to her own~\cite{Varian1974}, while its relaxations, {\em envy-freeness up to one item} (EF1)~\cite{LiptonMMS04,budish2011} and {\em envy-freeness up to any item} (EFX), accommodate the indivisibility of items by allowing envy to be eliminated through the removal of an appropriate item~\cite{CaragiannisKMP016}. {\em Proportionality} requires that each agent receive a share commensurate with an equal division of the total value or cost~\cite{Steinhaus}. Another prominent notion is the {\em maximin share} (MMS), which provides an agent with a fairness benchmark based on what she could guarantee herself by partitioning the items into bundles and receiving the least favorable bundle~\cite{budish2011}. These criteria capture different aspects of fairness and have motivated extensive research on the existence, computation, and approximation of fair allocations, particularly when items are indivisible and exact fairness guarantees may be unattainable.
	
	The classical maximin-share criterion implicitly assumes that all agents have equal entitlements. In many applications, however, agents may have different rights, obligations, or capacities, making equal treatment inappropriate. The {\em weighted maximin share} (WMMS) extends the MMS benchmark to settings in which agents are assigned positive weights representing their relative entitlements or responsibilities~\cite{FGH19}. Informally, in the allocation of indivisible chores, an agent evaluates her weighted maximin share by considering how she would partition all chores among agents with different weights, anticipating that she might receive the bundle that is most costly relative to its recipient's weight. This benchmark reflects both the agent's subjective costs and the asymmetry in agents' entitlements. When all weights are equal, WMMS reduces to the classical MMS notion. It is also closely related to {\em weighted proportionality}~\cite{SUKSOMPONG2025106519}, which assigns each agent a target cost proportional to her weight, and to weighted variants of envy-freeness. Unlike proportionality, however, WMMS explicitly accounts for indivisibility and the combinatorial structure of the chores. Moreover, while envy-freeness compares an agent's allocation with those of other agents, WMMS provides an individual fairness guarantee defined independently of the allocation ultimately chosen.
	
	The study of WMMS is particularly well motivated by practical settings in which undesirable tasks must be distributed among participants with unequal responsibilities. For example, in workplace scheduling, employees may have different contractual workloads, working hours, or capacities, so that assigning everyone the same amount of work would not necessarily constitute a fair outcome. Similarly, the distribution of teaching duties among faculty members may need to account for differences in appointments, administrative responsibilities, or teaching commitments. Other examples include assigning maintenance tasks among organizations with different resource capacities, distributing computational jobs among heterogeneous participants, and allocating household chores among individuals with different agreed-upon responsibilities. In such settings, indivisibility is often unavoidable: a task, shift, or assignment may have to be allocated in its entirety to a single participant. Furthermore, agents may experience different costs for performing the same task. WMMS offers a natural framework for incorporating these heterogeneous costs and unequal entitlements into a common fairness benchmark. 
	
	Beyond the existence of fair allocations, their efficient computation is a central concern. Even when an allocation satisfying a desired fairness criterion is guaranteed to exist, finding such an allocation may be computationally intractable. In particular, the indivisibility of chores and the heterogeneity of agents' costs create combinatorial difficulties, while unequal weights introduce additional challenges in balancing individual fairness guarantees. These obstacles motivate the study of {\em approximate WMMS} allocations, in which each agent's incurred cost is allowed to exceed her WMMS benchmark by a bounded multiplicative factor. The objective is to determine the smallest approximation factor that can be guaranteed for every instance and, equally importantly, the best factor achievable by a polynomial-time algorithm~\cite{FGH19,ACL19,AzizRaucheckerSchryenWalsh2017,Aziz2016,WLL24,LWX26,SS25}. These two questions are conceptually distinct: an existential guarantee establishes the intrinsic limits of fairness, whereas an efficient algorithm determines whether such guarantees can be realized in large-scale applications. Complementary impossibility results, obtained by constructing instances for which no allocation achieves a better approximation factor, help identify the fundamental barriers imposed by indivisibility and unequal entitlements. Narrowing the gap between achievable approximation guarantees and impossibility bounds is thus a central objective in the algorithmic study of weighted fair division of chores.

	This paper aims to narrows both sides of that gap. The upper-bound proof does not round the weights or the item costs to a geometric grid, and it does not require a common ordering of the original valuation rows. The lower-bound construction, in contrast, is deliberately structured: its rows have one common ordering, all item costs are positive, and its number of chores equals its number of agents.
	
	\section{Related work and Our Contribution}
	
	\subsection{Related work}\label{sec:related}
	
	Despite the recent surge of research on MMS allocations of indivisible goods and chores (see, e.g.,~\cite{KurokawaPW18,AmanatidisMNS17,BarmanKrishnamurthy2017,GhodsiEtAl2018,AmanatidisBFV22,GuoLiDeng2023,AHannaneh2024,GargTaki2021,AkramiGargSharmaTaki2023}), relatively little is known about the weighted setting, in which agents have unequal entitlements or weights~\cite{Suk25}.
	
	\paragraph{Weighted MMS for goods.}
	Farhadi et al.~\cite{FGH19} were the first to study the notion of weighted maximin share (WMMS). For goods, they constructed an instance showing that no allocation can guarantee better than an $n$-approximation to WMMS. They complemented this impossibility result with a simple ordinal algorithm achieving a matching $n$-approximation guarantee; namely, every agent receives a bundle whose value is at least $1/n$ of her WMMS.
	
	
	\paragraph{Weighted MMS for chores.}
	It is not immediate whether results for goods in one setting carry over to the other~\cite{Aziz2016}. Aziz, Chan, and Li~\cite{ACL19} established the existence of an $n$-approximate WMMS allocation and exhibited a lower bound of $4/3$. They also introduced the notion of optimal WMMS, which benchmarks an algorithm against the smallest common relaxation factor for which an allocation exists on the given instance. An approximation guarantee for this instance-dependent benchmark is conceptually different from a universal upper bound on the relaxation factor itself. Wang, Li, and Lu~\cite{WLL24} obtained an $O(\log n)$ approximation guarantee and determined the exact worst-case ratio for two agents as $(1+\sqrt{3})/2$. More recently, Li, Wang, and Xing~\cite{LWX26} proved an upper bound of $12$ together with an asymptotic lower bound of $2$, which constitute the most immediate benchmarks for our results. Independently, Seddighin and Seddighin~\cite{SS25} established a constant-factor guarantee of $20$.
	
	\begin{table}[t]
		\centering
		\small
		\begin{tabular}{@{}p{0.49\linewidth}cc@{}}
			\toprule
			Work & Existence upper bound & Universal lower bound\\
			\midrule
			Aziz--Chan--Li~\cite{ACL19} & $n$ & $4/3\approx 1.333$\\
			Wang--Li--Lu~\cite{WLL24} & $O(\log n)$ & $(1+\sqrt3)/2\approx 1.366$\\
			Li--Wang--Xing~\cite{LWX26} & $12$ & $2$\\
			Seddighin--Seddighin~\cite{SS25} & $20$ & ---\\
			{\bf This work} &  $\pmb{4.124234528}$ & $\pmb{e\approx 2.71828}$\\
			\bottomrule
		\end{tabular}
		\caption{Results for approximating WMMS allocations with chores.}
		\label{tab:comparison}
	\end{table}

	\subsection{Our contribution}
	
	We studied the allocation of indivisible chores among agents with arbitrary positive weights and additive nonnegative costs. Our contribution are as followss. 
	\begin{enumerate}
		\item We prove that every instance admits a $4.124234528$-WMMS allocation, and that a $4.12424$-WMMS allocation can be computed in polynomial time for rational input. The upper bound follows from a lossless row-wise packing normalization, a restricted fractional assignment problem, and a continuum of dual-price inequalities. Integrating these inequalities cancels their scale-dependent correction terms; assignment rounding then incurs only the largest supported chore cost.
		\item In the opposite direction, we construct strictly positive, identically ordered instances whose optimal WMMS factors converge to $e$. The construction uses geometrically growing groups, nested costs, and integral row-wise packing certificates. It establishes an unconditional existential lower bound, not an $\np$-hardness-of-approximation result. 
	\end{enumerate}
	Together, the results place the optimal universal existence factor between $e$ and $4.124234528$. 


	\section{Preliminaries}\label{sec:preliminaries}
	A problem instance of fair allocation is $I=(\cN,\cM,w,v)$, where $\cN=[n]$ is the set of $n$ agents, $\cM=[m]$ is the set of $m$ indivisible chores (or items); each agent is associated with a positive {\em weight} $w_i$, and an {\em additive cost functions} $v_i:\cM\to\QQ_{\ge0}$. For any subset $X\subseteq \MM$, $v_i(X)=\sum_{h\in X}v_{ih}$, where $v_{ih}$ denotes the cost of chore $h$ for the agent $i$.  Unless stated otherwise, we assume $\sum_{i\in\NN}w_i=1$, and denote  $w_*:=\max_{j\in\cN}w_j$ as the maximum weight among agents. An {\em integral allocation} (or {\em assignment}) is a complete partition $A=(A_i)_{i\in\cN}$ of $\cM$ into $n$ disjoint subsets. We may use a vector  $\bx=(x_{ih})_{i\in\cN,h\in\cM}$ to present a {\em fractional  allocation}, where $x_{ih}\in\mathbb{Q}_{\ge 0}$ is the fraction of the chore $h$ assigned to agent $i$, such that $\sum_{i\in\cN}x_{ih}=1$ for every chore $h\in\cM$.
	\begin{definition}[Weighted maximin share (WMMS)]
		The weighted maximin share of an agent $i$ is defined as
		\begin{equation}\label{eq:wmms-intro}
			\mu_i=\WMMS_i
			=w_i\min_{(A_1,\ldots,A_n)}\max_{j\in\NN}\frac{v_i(A_j)}{w_j},
		\end{equation}
		where the minimum ranges over complete integral partitions $A$.
	\end{definition} 
	The agent evaluates every bundle using her own costs, but each bundle is compared with its designated recipient's weight. Computing the weighted maximin share $\WMMS_i$ is known to be $\np$-hard but can be well approximated~\cite{HS88}. In fact,  computing $\WMMS_i$ is exactly equivalent to the minimum makespan  problem on uniform machines and thus can be approximated in polynomial time to within any factor of $1+\delta$, for any fixed $\delta\in(0,1)$.

	The defining minimum (\ref{eq:wmms-intro}) is attained because the instance $I$ has finitely many $O(n^m)$ integral partitions. An allocation in which every agent incurs a cost no greater than her weighted maximin share (WMMS) does not always exist, even in the unweighted setting. This motivates the study of approximate WMMS allocations, which relax the fairness requirement by allowing each agent's cost to exceed her WMMS by a bounded multiplicative factor.
	\begin{definition}[Approximate $\alpha$-$\WMMS$]
		For $\alpha\ge 1,$ an allocation $A$ is $\alpha$-WMMS if $v_i(A_i)\le\alpha\mu_i$ for every agent $i\in\cN$. The value of $\alpha$ may depend on the size of the problem instance.
	\end{definition} 
	We may assume that $v_i(\mathcal M)>0$ for every agent $i$. Indeed, if $v_i(\mathcal M)=0$ for some agent $i$, assigning all chores to this agent and leaving the other agents with empty bundles yields an exact WMMS allocation. Our goal is to determine the smallest constant $\alpha\geq 1$ such that every instance admits an $\alpha$-WMMS allocation. To this end, we establish several useful properties of the normalized cost functions, which will serve as key ingredients in deriving a universal approximation guarantee.
	

	\begin{lemma}[Normalized cost function]\label{lem:normalization}
		Suppose every $\mu_i>0$, and define the additive cost function $c_i$ with $c_{ih}=w_iv_{ih}/\mu_i$. Then, for each $i$ there is an integral assignment $(B_j^i)_{j\in\NN}$ such that
		\begin{equation}\label{eq:packing}
			c_i(B_j^i)\le w_j,\qquad(j\in\NN).
		\end{equation}
		We call $(B^i_j)_{j\in\mathcal N}$ the row-$i$ packing.  In particular, $c_i(\MM)\le1$ and $c_{ih}\le w_*$ for every chore $h$ and agent $i$. An integral assignment with $c_i(A_i)\le Cw_i$ for every $i$ gives a $C$-$\WMMS$ allocation for the original instance.
	\end{lemma}
	\begin{proof}
		Take a WMMS-defining partition in row $i$ and multiply its inequalities by $w_i/\mu_i$. Summing the resulting inequalities gives the total-cost bound; each item lies in a bin of capacity at most $w_*$. Finally, $v_i(A_i)=(\mu_i/w_i)c_i(A_i)\le C\mu_i$.
	\end{proof}
	
	The upper-bound proof works for \emph{any} nonnegative matrix satisfying \cref{eq:packing}, even if its normalized costs are strictly smaller than its weights. We refer to this hypothesis as the \emph{packing promise}. It is stronger than the scalar inequalities $c_i(\MM)\le1$, since the witnesses are whole-item packings. The proof uses their existence, not a common partition or an ordering shared across rows.
	
	If $\mu_i=0$ for some agent, every item costs zero to her. Indeed, a zero makespan with positive bin weights permits no positive-cost item. Assigning all chores to that agent solves the instance. Empty chore sets and the one-agent case are likewise immediate. We may therefore assume positive costs and a nonempty chore set whenever ratios are used.

	\section{A universal 4.124234528 upper bound}\label{sec:upper}
	
	We first introduce a function $\alpha:\Omega\subseteq\mathbb{R}^2\rightarrow \mathbb{R}_{>0}$, where $\Omega=\{(\rho,r)\in \mathbb{R}^2~:~1<\rho<r\}$, and
	\begin{equation}\label{eq:alpha-definition}
		\alpha(\rho,r)=\rho+
		\frac{1+\ln(r/\rho)}{(1+\ln\rho)/\rho-1/r}.	
	\end{equation}
	Let $\astar=\min_{(\rho,r)\in \Omega}\alpha(\rho,r).$
	\begin{lemma}\label{lem:compute:alpha}
		$\astar=4.124234527128370976\ldots$
	\end{lemma}
	Our main result is stated as follows.
	\begin{theorem}[Universal upper bound]\label{thm:upper-main}
		Every finite additive nonnegative chore instance with positive weights admits an $\alpha(\rho,r)$-WMMS allocation for every $r>\rho>1$. In particular, for  every fixed $\varepsilon>0$, an $(\astar+\varepsilon)$-WMMS allocation can be computed in polynomial time. Explicit rational parameters give a polynomial-time $4.12424$ guarantee.
	\end{theorem}

	\paragraph{An overview of the Algorithm}
	
	Suppose that the cost functions are normalized so that they admit an integral packing into bins whose capacities correspond to the agents' weights. We  introduce an eligibility threshold $\rho>1$ and restrict assignments to agent–chore pairs whose normalized costs do not exceed $\rho$ times the corresponding agent's weight. On these eligible pairs, we formulate a fractional assignment linear program that minimizes the maximum normalized load across agents. By linear programming duality, its feasibility reduces to establishing an inequality between the total prices of chores and the weighted sum of dual multipliers. Using the row-wise packing property, we show that suitably capped dual prices can be packed into appropriately scaled capacities, yielding a family of inequalities across different thresholds. We further develop a multiscale analysis that combines these inequalities while retaining their scale-dependent correction terms. By integrating over a continuum of scales and carefully combining the resulting inequalities, the correction terms cancel, leading to an improved bound on the fractional assignment loads. We then apply assignment rounding, which increases each agent's normalized load by at most the eligibility threshold $\rho$. Finally, optimizing the parameters yields a universal $4.124234528$-WMMS existence guarantee and a polynomial-time $4.12424$-WMMS approximation algorithm for rational inputs. The main steps are identified by Algorithm~\ref{alg:wmms}. Its correctness and thus the proof of Theorem~\ref{thm:upper-main} is presented in the following sections.

	\begin{algorithm}
		\caption{A polynomial-time $4.12424$-WMMS allocation algorithm}
		\label{alg:wmms}
		\begin{algorithmic}[1]
			\Require Agents $\cN$, chores $\cM$, positive weights
			$(w_i)_{i\in \cN}$ with $\sum_{i\in \cN}w_i=1$, additive costs $(v_i)_{i\in \NN}$
			\Ensure An allocation $A=(A_i)_{i\in \cN}$ satisfying
			$v_i(A_i)\le 4.12424\cdot\WMMS_i$ for all $i$
			
			\State Set $\delta\leftarrow 10^{-6}$, $\rho\leftarrow 689/625$,
			and $\beta\leftarrow 604367/200000$.
			
			\If{$v_i(\MM)=0$ for some agent $i$}
			\State Assign all chores to agent $i$.
			\State \Return the resulting allocation.
			\EndIf
			
			\Statex \textbf{Phase I: Approximate the WMMS benchmarks}
			
			\For{each agent $i\in \cN$}
			\State Apply the $(1+\delta)$-PTAS for makespan
			minimization on uniformly related machines,
			using job sizes $v_{ih}$ and machine speeds $w_j$.
			\State Let $\widehat T_i$ be the makespan of
			the returned schedule.
			\State Set $\widehat\mu_i\gets w_i\widehat T_i$.
			\Comment{$\mu_i\le\widehat\mu_i\le(1+\delta)\mu_i$}
			\EndFor
			
			\Statex \textbf{Phase II: Normalize costs and select eligible edges}
			
			\For{each agent $i\in \cN$ and chore $h\in \cM$}
			\State Set
			$c_{ih}\gets \dfrac{w_i v_{ih}}{\widehat\mu_i}$.
			\EndFor
			
			\State Define $E_\rho\leftarrow \{(i,h)\in \cN\times \cM:
			c_{ih}\le\rho w_i\}.$

			\Statex \textbf{Phase III: Compute a fractional assignment}
			
			\State Find a feasible solution $(x_{ih})_{(i,h)\in E_\rho}$
			to the following linear program:
			\[
			\begin{array}{ll}
				\displaystyle
				\sum_{i:(i,h)\in E_\rho}x_{ih}=1,
				& \forall h\in \cM,\\[2mm]
				\displaystyle
				\sum_{h:(i,h)\in E_\rho}c_{ih}x_{ih}\le \beta w_i,
				& \forall i\in \cN,\\[2mm]
				x_{ih}\ge0,
				& \forall(i,h)\in E_\rho.
			\end{array}
			\]
			
			\Statex
			\Statex \textbf{Phase IV: Round the fractional assignment}
			
			\For{each agent $i\in \cN$}
			\State Sort chores with $x_{ih}>0$ in
			nonincreasing order of $c_{ih}$.
			\State Let
			$q_i\gets\left\lceil\sum_{h\in \cM}x_{ih}\right\rceil$.
			\State Create $q_i$ slots of unit capacity.
			\State Place fractional chore masses $x_{ih}$
			consecutively into these slots,
			following the sorted order.
			\EndFor
			
			\State Construct a bipartite graph $\cH$ between
			chores and slots, connecting a chore
			to each slot containing a positive fraction of it.
			
			\State Find a matching in $\cH$ covering every chore.
			
			\For{each chore $h\in \cM$}
			\State Assign $h$ to the agent owning
			its matched slot.
			\EndFor
			
			\State \Return the resulting integral allocation
			$A=(A_i)_{i\in \cN}$.
			
		\end{algorithmic}
	\end{algorithm}
	
	\subsection{Establishing the existence of the $\alpha$-WMMS allocation}
	
	\subsubsection{The dual-price formulation}
	Let $c_i$ be the normalized cost function of agent $i$ defined as in Lemma~\ref{lem:normalization}. An edge $(i,h)$ is called \emph{eligible} if $c_i(h)\le\rho w_i$.  Fix $(\rho,r)\in\Omega$, and let
	\[
	E_\rho=\{(i,h)\in \cN\times \cM:
	c_{ih}\le\rho w_i\},
	\]
	and consider the following restricted assignment $\LP$ (depending on a given value $\beta>0$), which assigns a (fractional) variable to each {\em eligible pair} in $E_\rho$.
	\[
	\begin{array}{ll}
		\displaystyle
		\sum_{i:(i,h)\in E_\rho}x_{ih}=1,
		& \forall h\in \cM,\\[2mm]
		\displaystyle
		\sum_{h:(i,h)\in E_\rho}c_{ih} x_{ih}\le \beta w_i,
		& \forall i\in \cN,\\[2mm]
		x_{ih}\ge0,
		& \forall(i,h)\in E_\rho.
	\end{array}
	\]
	
	It is observed that there is at least one eligible recipient for every chore.  To see this, choose an agent $a$ of maximum weight, so $w_a=w_*$.
	By \cref{eq:packing}, every chore belongs to some bundle
	in row $a$'s witnessing partition and therefore $c_{ah}\le w_j\le w_* =w_a<\rho w_a.$ 
	Consequently $(a,h)\in E_\rho$ for every $h$.  This elementary
	observation ensures that the covering constraints themselves are not
	infeasible merely because some chore has no eligible recipient.

	We first show the feasibility of $\LP$. Rather than explicitly constructing a feasible fractional assignment $\bx=(x_{ih})$, we establish its existence through linear programming duality. Specifically, by Farkas' lemma, the fractional assignment LP is feasible if and only if no dual certificate of infeasibility exists. Consequently, it suffices to show that every admissible dual solution satisfies some inequality, as formalized in the following lemma.
	
	\begin{lemma}[Dual feasibility criterion]\label{lem:dual}
		For a given $\beta\ge 0$, the $\LP$ is feasible if and only if every choice $\lambda_i\ge0$ and $p_h\ge0$ satisfying
		\begin{equation}\label{eq:dual}
			p_h\le\lambda_i c_{ih},\qquad\text{for }(i,h)\in E_\rho,
		\end{equation}
		obeys
		\begin{equation}\label{eq:dual-goal}
			P\le \beta W,\qquad \text{where}\qquad P=\sum_{h} p_h,\qquad W=\sum_{i}\lambda_iw_i.
		\end{equation}
	\end{lemma}
	\begin{proof}
		\emph{Necessity.} Suppose $\LP$ has a feasible solution $x$.  Multiply
		$p_h\le\lambda_i c_{ih}$ by $x_{ih}\ge0$, and sum over eligible
		edges:
		\begin{align*}
			\sum_h p_h
			=\sum_h p_h\sum_{i:(i,h)\in E_\rho}x_{ih}
			\le \sum_i\lambda_i
			\sum_{h:(i,h)\in E_\rho}c_{ih}x_{ih}
			\le\sum_i\lambda_i \beta w_i=\beta W.
		\end{align*}
		The first equality uses coverage, while the last inequality uses the
		load constraints and the nonnegativity of the multipliers.
		
		\emph{Sufficiency.} Consider the following always-feasible auxiliary
		linear program, with a nonnegative load parameter $z$:
		\[
		\begin{array}{ll}
			\text{minimize} &z\\
			\text{subject to}&\displaystyle\sum_{i:(i,h)\in E_\rho}x_{ih}=1
			\quad(h\in\cM),\\
			&\displaystyle\sum_hc_{ih}x_{ih}\le zw_i\quad(i\in\cN),\\
			&\bx\ge0,\quad z\ge0.
		\end{array}
		\]
		It is feasible because each chore has an eligible edge, and sufficiently
		large $z$ covers the resulting finite loads.  Its optimum $z^*$ is
		finite.  The dual can be written as
		\[
		\begin{array}{ll}
			\text{maximize}&\displaystyle\sum_hp_h\\
			\text{subject to}&p_h\le\lambda_i c_{ih},\qquad ~~(i,h)\in E_\rho,\\
			&\displaystyle\sum_i\lambda_iw_i\le1,
			\qquad\lambda_i\ge0,
		\end{array}
		\]
		where initially the equality-constraint multipliers $p_h$ are
		unrestricted in sign.  Because all costs and $\lambda_i$ are
		nonnegative, replacing a negative $p_h$ by zero preserves the dual
		constraints and cannot decrease its objective.  We may therefore
		restrict attention to $p_h\ge0$.
		
		If every feasible price/multiplier pair satisfies $P\le \beta W$, every
		dual solution with $W\le1$ has $P\le \beta$.  Linear-programming strong
		duality implies $z^*\le \beta$.  An optimal solution at $z^*$ is then
		feasible to $\LP$ with the prescribed budget $\beta$.
		The argument also covers the degenerate cases $P=W=0$ and zero
		multipliers.
	\end{proof}

	Fix an arbitrary feasible dual pair $(\lambda=(\lambda_i)_i,p=(p_h)_h)$ satisfying (\ref{eq:dual}). By Lemma~\ref{lem:dual}, to establish the feasibility of the fractional assignment $\LP$, we aim to derive an upper bound on the total dual price $P=\sum_h p_h$ in terms of the weighted sum $W=\sum_i\lambda_iw_i$ of the dual multipliers. Rather than analyzing the individual prices directly, we examine their distribution above different thresholds. The motivation is that chores with sufficiently large prices, after appropriate capping, can only be accommodated by bins with sufficiently large capacities in the WMMS witnessing partitions. By expressing the total capped prices through upper-tail functions (Lemma~\ref{d37:lemma-tail}), we can translate these packing restrictions into inequalities involving the dual variables (Lemma~\ref{lem:allscale}). These inequalities will subsequently be integrated over a continuum of scales to establish the desired bound $P\le \beta W$ (Theorem~\ref{thm:fractional}), thereby proving the feasibility of the fractional assignment $\LP$.
	
	\subsubsection{Represent the dual prices by their upper tails}
	\begin{definition}[Upper-tail functions]\label{def:up-tail-func}
		For $z\ge0$, define $D(z)=|\{h:p_h>z\}|$, and $E(z)=\sum_{h}(p_h-z)^+$, where $x^+=\max\{x,0\}$, respectively the number of chores whose prices exceed the threshold $z$ and the total excess of these prices above $z$.
	\end{definition}
	In particular,
	\begin{equation}\label{d37:E-boundary}
		E(0)=P,\qquad E(z)=0\quad\text{if }z\ge p_{\max}:=\max_{h\in\mathcal M}p_h.
	\end{equation}
	
	\begin{lemma}[Upper-tail identities]\label{d37:lemma-tail}
		For all $x,y\ge0$, with oriented integration if $x>y$,
		\begin{equation}\label{d37:tail-integral}
			E(x)-E(y)=\int_x^y D(z)\,\mathrm dz.
		\end{equation}
		In particular, $E(z)=\int_z^\infty D(u)\,\mathrm du$.
		For $0\le b<c$ one also has
		\begin{equation}\label{d37:capped-identity}
			\sum_{h:p_h>b}\min\{p_h,c\}
			=E(b)-E(c)+bD(b).
		\end{equation}
	\end{lemma}
	\begin{proof}
		For a fixed chore $h$, the indicator $\mathbf1_{\{p_h>z\}}$
		is equal to $1$ on $[0,p_h)$ and $0$ beyond $p_h$.  Hence
		\[
		\int_x^y\mathbf1_{\{p_h>z\}}\,\mathrm dz
		=(p_h-x)_+-(p_h-y)_+.
		\]
		Summing over the finitely many chores proves \eqref{d37:tail-integral}.
		The representation $E(z)=\int_z^\infty D(u)\,\mathrm du$
		follows by taking $y\ge p_{\max}$.
		
		To prove \eqref{d37:capped-identity}, consider a chore with $p_h>b$.
		If $p_h\le c$, then
		\[
		(p_h-b)_+-(p_h-c)_++b=(p_h-b)+b=p_h.
		\]
		If $p_h>c$, then the same expression equals
		$(p_h-b)-(p_h-c)+b=c$.  Thus in either case
		\[
		\min\{p_h,c\}=(p_h-b)_+-(p_h-c)_++b.
		\]
		For $p_h\le b$ both the first two positive-part terms vanish, so
		summing the displayed equality only over $p_h>b$ yields
		$E(b)-E(c)+bD(b)$.
	\end{proof}
	
	\subsubsection{Threshold profiles of dual multipliers and weights}
	For $t>0$, define
	\begin{equation}\label{eq:profiles}
		\omega(t)=\max\{w_i:\lambda_i\le t\},
		\qquad
		R(t)=\sum_{i:\lambda_i>t}w_i,
	\end{equation}
	where the maximum of an empty set is defined to be zero.
	In other words, $\omega(t)$ denotes the largest weight
	among agents whose dual multipliers are at most $t$,
	while $R(t)$ represents the total weight of agents
	whose dual multipliers exceed $t$.
	Clearly, $\omega(t)$ is nondecreasing and $R(t)$
	is nonincreasing in $t$.
	Moreover, these functions satisfy
	\begin{equation}\label{eq:profile-facts}
		\sum_{i:w_i>\omega(t)}w_i\le R(t)
		\quad\text{and}\quad
		\int_0^\infty R(t)\,\dd t
		=\sum_i\lambda_iw_i=W.
	\end{equation}
	Indeed, $w_j>\omega(t)$ implies $\lambda_j>t$; the integral identity follows by integrating the indicators $\mathbf1_{\{t<\lambda_i\}}$.
	
	\begin{lemma}[Capped prices admit integral packings]\label{lem:capped}
		For any $u>0$ such that $\omega(u)>0$, the capped sizes of the indivisible chores $z_h(u)=\min\{p_h,\rho u\omega(u)\}$ 
		can be packed into bins of capacities $uw_1,\ldots,uw_n$.
	\end{lemma}
	\begin{proof}
		Choose agent $a$ with $w_a=\omega(u)$ and $\lambda_a\le u$. If $(a,h)$ is eligible, then $p_h\le\lambda_a c_{ah}\le uc_{ah}$. Otherwise, $uc_{ah}>\rho u w_a$. Thus $z_h(u)\le uc_{ah}$ for all chores, including those not eligible
		for agent $a$.  Take a complete integral partition witnessing
		\cref{eq:packing} for row-$a$ partition.  In its $j$th bundle,
		\[
		\sum_{h\in B_j^a}z_h(u)
		\le u\sum_{h\in B_j^a}c_{ah}
		\le uw_j.
		\]
		The \emph{same} partition therefore packs all capped prices
		integrally.
		~\end{proof}
	
	\subsubsection{Turn the packing into a pointwise price inequality}
	For every scale $q>0$, denote
	\begin{equation}\label{eq:corrections}
		F_q(t)=E(q t\omega(t)),\qquad
		N_q(t)=q t\omega(t)D(q t\omega(t)),\qquad
		J(q)=\int_0^\infty N_q(t)\frac{\dd t}{t}.
	\end{equation}
	Note that $J(q)$ is finite and nonnegative. It aggregates the threshold contributions of expensive chores over all dual-multiplier thresholds $t>0$, for a fixed scale $q$. In addition, near zero, $N_q(t)\le qmw_*t$, and for sufficiently large $t$, $\omega(t)=w_*$ and $N_q(t)=0$. The estimates are uniform on compact positive scale intervals. Also, $F_q(t)\to P$ as $t\downarrow0$, and $F_q(t)=0$ eventually. Thus compact scale integrations below are justified by nonnegativity or absolute integrability.
	
	\begin{lemma}[All-scale inequality]\label{lem:allscale}
		For every $s\ge1$,
		\begin{equation}\label{eq:allscale}
			P\ln s+J(s)-\int_\rho^s\frac{J(q)}q\dd q\le sW.
		\end{equation}
		The integral has its usual orientation when $s<\rho$.
	\end{lemma}
	\begin{proof}
		We first show the pointwise inequality
		\begin{equation}\label{eq:pointwise}
			F_s(t)-F_\rho(st)+N_s(t)\le stR(t).
		\end{equation}
		Put $u=st$. If $\omega(u)=0$, both relevant profiles are zero and the left side is zero. Otherwise use \cref{lem:capped} at $u$. Every item with $p_h>u\omega(t)$ has capped size strictly greater than $u\omega(t)$, because $\rho u\omega(u)>u\omega(t).$ 
		Here $\omega(u)\ge\omega(t)$, $s\ge1$, and $\rho>1$. Therefore these items occupy only bins with $w_j>\omega(t)$. Their total available capacity is at most $uR(t)$ by \cref{eq:profile-facts}. The sum of their capped sizes equals
		\begin{align*}
			\sum_{h:p_h>u\omega(t)}\min\{p_h,\rho u\omega(u)\}
			&=E(u\omega(t))-E(\rho u\omega(u))
			+u\omega(t)D(u\omega(t))\\
			&=F_s(t)-F_\rho(st)+N_s(t).
		\end{align*}
		This also handles $\omega(t)=0<\omega(u)$: every selected positive item uses a positive-capacity bin.
		
		Integrate \cref{eq:pointwise} against $\dd t/t$. Logarithmic telescoping gives
		\begin{equation}\label{eq:log-telescoping}
			\int_0^\infty\bigl(F_\rho(t)-F_\rho(st)\bigr)\frac{\dd t}{t}=P\ln s.
		\end{equation}
		For $s>1$, integration from $\varepsilon$ to $T$ and substitution in the second term give the integral of $F_\rho(t)/t$ over $[\varepsilon,s\varepsilon]$ minus its integral over $[T,sT]$. The endpoint values give the identity. The case $s=1$ is immediate.
		
		Furthermore, $E(x)-E(y)=\int_x^yD(z)\dd z$ implies
		\begin{equation}\label{eq:scale-exchange}
			\int_0^\infty\bigl(F_s(t)-F_\rho(t)\bigr)\frac{\dd t}{t}
			=-\int_\rho^s\frac{J(q)}q\dd q.
		\end{equation}
		Use $z=qt\omega(t)$ and interchange the compact scale integral with the $t$-integral. Combining \cref{eq:log-telescoping,eq:scale-exchange}, the definition of $J(s)$, and \cref{eq:profile-facts} proves \cref{eq:allscale}.
	\end{proof}
	
	\subsubsection{Cancellation of the correction terms}
	\begin{theorem}[Multiscale fractional bound]\label{thm:fractional}
		The assignment system $\LP$ is feasible with
		\begin{equation}\label{eq:fractional-bound}
			\beta=\beta(\rho,r)=\frac{1+\ln(r/\rho)}{(1+\ln\rho)/\rho-1/r}.
		\end{equation}
	\end{theorem}
	\begin{proof}
		Let
		\[
		Q(s)=\int_\rho^s\frac{J(q)}q\dd q.
		\]
		Then $Q(\rho)=0$ and $Q'(s)=J(s)/s$ almost everywhere. Divide \cref{eq:allscale} by $s^2$, integrate from $\rho$ to $r$, and add \cref{eq:allscale} at $s=r$, divided by $r$. The correction integral is
		\[
		\int_\rho^r\frac{J(s)-Q(s)}{s^2}\dd s
		=\int_\rho^r\left(\frac{Q(s)}s\right)'\dd s
		=\frac{Q(r)}r.
		\]
		Adding the endpoint cancels $Q(r)$ and leaves $J(r)/r$. Also,
		\[
		\int_\rho^r\frac{\ln s}{s^2}\dd s+\frac{\ln r}{r}
		=\frac{1+\ln\rho}{\rho}-\frac1r,
		\qquad
		\int_\rho^r\frac{\dd s}s+1=1+\ln(r/\rho).
		\]
		Consequently,
		\begin{equation}\label{eq:master-upper}
			P\left(\frac{1+\ln\rho}{\rho}-\frac1r\right)
			+\frac{J(r)}r
			\le W\left(1+\ln\frac r\rho\right).
		\end{equation}
		The coefficient of $P$ is positive because $r>\rho>1$. Discard the nonnegative endpoint term and apply \cref{lem:dual} with $\alpha=\alpha(\rho,r)$ given by \cref{eq:fractional-bound}. Hence, the assignment system $\LP$ is feasible and the proof is completed.
	\end{proof}
	
	\subsection{Computing an $(\alpha+\epsilon)$-WMMS allocation}
	We make use of the standard matching argument, which is 
	included here to make the additive error precise.
	
	\begin{lemma}[Additive assignment rounding~\cite{ST93}]\label{lem:rounding}
		Suppose a fractional assignment $\bx$ covers each chore exactly once,
		the fractional cost of agent $i$ is at most $L_i$, and every edge with
		$x_{ih}>0$ has cost at most $b_i$.  One can compute in polynomial time an integral
		assignment using only supported edges with the (integral) cost of agent $i$ is at most
		$L_i+b_i$, simultaneously for all agents. 
	\end{lemma}
	\begin{proof}
		Fix an agent $i$.  List the chores with $x_{ih}>0$ in
		nonincreasing order of cost $c_i(h)$.  Regard $x_{ih}$ as an item
		\emph{mass}, and lay these masses consecutively into slots of mass at
		most one.  If a mass crosses a slot boundary, split that fractional
		mass between the adjacent slots.  Every slot is full, except
		possibly the last.  Denote the slots of agent $i$ by
		$(i,1),\ldots,(i,q_i)$; an agent with no positive fractional mass
		has no slots.
		
		Construct a bipartite graph $\cH$ with chores on the left and all the slots
		on the right.  Join chore $h$ to slot $(i,k)$ precisely when a
		positive portion of its fractional mass was placed there.
		Each chore has total fractional mass one over its incident slots,
		and each slot has total incident fractional mass at most one.
		This is a fractional matching saturating every chore.  The
		integrality of bipartite matching (equivalently, a max-flow
		integrality theorem) therefore gives an \emph{integral} matching
		saturating every chore.  Assign a chore to the owner of its matched
		slot.  Every chore is assigned once, and only supported assignment
		edges are used.
		
		It remains to bound the load of one agent.  The cost of the item
		matched to her first slot, if any, is at most $b_i$.  Suppose a
		chore $h$ is matched to slot $(i,k)$ with $k\ge2$.  Because the
		fractional pieces were ordered by decreasing cost, \emph{every}
		positive fractional piece in the preceding slot $(i,k-1)$ has
		cost at least $c_{ih}$.  That preceding slot is full, of total
		mass one.  Its fractional cost is therefore at least $c_{ih}$.
		
		Different matched slots $k\ge2$ charge to different preceding
		slots $k-1$, and none of these charges exceeds the corresponding
		preceding slot's fractional cost.  Summing their charges is at most
		the total original fractional load of agent $i$, hence at most
		$L_i$.  Adding the possible first-slot item gives integral load
		at most $L_i+b_i$.  The argument holds simultaneously for every
		agent because the one matching assigns all items consistently.
		
		Finally, the slot graph has polynomial size, and bipartite matching is
		polynomial time. This completes the proof.
	\end{proof}

	Apply \cref{lem:rounding} with $L_i=\alpha(\rho,r)w_i$ and $b_i=\rho w_i$. Then
	\[
	c_i(A_i)\le\bigl(\beta(\rho,r)+\rho\bigr)w_i=\alpha(\rho,r) w_i.
	\]
	\Cref{lem:normalization} proves the existential part of \cref{thm:upper-main}. The main upper-bound theorem thus holds directly on the original matrix: no identical-ordering reduction or weight discretization is needed.

	For rational original input, use the related-machine scheduling scheme~\cite{HS88} separately for each row. Job $h$ has size $v_i(h)$ and machine $j$ has speed $w_j$. If the returned feasible makespan is $\widehat\tau_i$, then $\widehat\mu_i=w_i\widehat\tau_i$ satisfies
	\begin{equation}\label{eq:benchmark-estimate}
		\mu_i\le\widehat\mu_i\le(1+\delta)\mu_i
	\end{equation}
	for any fixed $\delta>0$. Normalize instead by $\widehat c_i(h)=w_iv_i(h)/\widehat\mu_i$. The returned schedules themselves witness \cref{eq:packing}. Consequently, \cref{thm:fractional} applies without knowing exact WMMS.
	
	Choose rational $\rho,r$ and a rational $\widehat \beta>\beta(\rho,r)$. Solve $\LP$ for $\widehat c$ with budget $\widehat \beta w_i$, then round. The original cost is at most
	\begin{equation}\label{eq:algorithm-factor}
		(\widehat \beta+\rho)\widehat\mu_i
		\le(\widehat \beta+\rho)(1+\delta)\mu_i.
	\end{equation}
	There are at most $nm$ assignment variables and $n+m$ main constraints. The WMMS approximation is polynomial time for fixed precision; solving the rational LP and finding the matching are polynomial-time operations. Nearby rational parameters and sufficiently small fixed errors prove the $(\astar+\varepsilon)$ assertion.
	
	For an explicit guarantee, take
	\begin{equation}\label{eq:algorithm-parameters}
		\rho=\frac{689}{625},\qquad r=\frac{15109}{5000},\qquad
		\widehat \beta=\frac{604367}{200000},\qquad \delta=10^{-6}.
	\end{equation}
	Exact logarithm enclosures establish $\alpha(\rho,r)<\widehat \alpha$. Moreover,
	\begin{equation}\label{eq:algorithm-exact}
		(\widehat \beta+\rho)(1+\delta)
		=\frac{824847824847}{200000000000}
		=4.124239124235<4.12424.
	\end{equation}
	This finishes the algorithmic part of \cref{thm:upper-main}. 
	
	\section{An existential lower bound of \texorpdfstring{$e$}{e}}\label{sec:lower}
	For a problem instance $I$ with positive weighted maximin shares, define
	\[
	\gamma(I)=\min_{(A_1,\ldots,A_n)}\max_{i\in \cN}\frac{v_i(A_i)}{\mu_i},
	\qquad \Gamma=\sup_I\gamma(I).
	\]
	Basically, $\Gamma$ is the optimal universal \emph{existence} factor, rather than an approximation ratio relative to the best allocation of a fixed instance.
	
	Our main result is as follows.
	\begin{theorem}[Existential lower bound]\label{thm:wmms-existence}
		For any fixed $\epsilon\in(0,e-1)$, there is an instance which does not admit an $(e-\epsilon)$-WMMS allocation, where $e\approx 2,71828$.
	\end{theorem}
	To get the result, we show a construction of an explicit instance as in the theorem below.
	\begin{theorem}\label{thm:lower-main}
		For every integer $t\ge4$ and $K\ge t+1$, there is an explicit  instance $I_{t,K}$ with $2^K-1$ agents and $2^K-1$ chores such that
		\begin{equation}\label{eq:main-sandwich}
			\left(1+\frac1t\right)^{t-3}\frac{K}{K+t-1}
			\ \le\ \gamma(I_{t,K})
			\ \le\ \left(1+\frac1t\right)^{t-1}.
		\end{equation}
		All costs are strictly positive, additive, and identically ordered. 
	\end{theorem}
	By taking $K=t^2$, both endpoints of \cref{eq:main-sandwich} converge to $e$:
	\[
	\left(1+\frac1t\right)^{t-3}\frac{t^2}{t^2+t-1}\longrightarrow e,
	\qquad
	\left(1+\frac1t\right)^{t-1}\longrightarrow e.
	\]
	Hence, for every $0<\varepsilon<e-1$, a sufficiently large finite $t$ therefore gives an instance admitting no $(e-\varepsilon)$-WMMS allocation. This gives \cref{thm:wmms-existence}.
	
	It therefore suffices to prove Theorem~\cref{thm:lower-main}.
	
	\paragraph{An overview of the instance construction. }
	To establish our lower bound, we construct a family of weighted
	chore-allocation instances with a layered structure, consisting of
	$K$ groups of agents and $K$ corresponding blocks of chores, where
	both the number of agents and the number of chores in layer $i$
	equal $2^{i-1}$. The agents are assigned carefully chosen
	geometrically structured weights, while their additive cost
	functions follow a nested pattern in which all agents agree on
	the ordering of chores and every chore has strictly positive cost.
	The key feature of the construction is that each agent's weighted
	maximin share can be bounded from above by explicitly constructing
	an integral packing of the chores into bins whose capacities
	correspond to the agents' weights. These benchmark bounds force
	agents, under a sufficiently small approximation factor, to receive
	chores only from certain later layers, allowing us to derive a
	contradiction by comparing the total cost of any allocation with
	the sum of the permitted WMMS budgets. By appropriately choosing
	the parameters and letting the number of layers grow, we show
	that the resulting lower bound converges to $e$, thereby proving
	that no universal WMMS approximation factor smaller than $e$
	is possible.
	
	

	\subsection{Agent groups, weights, and additive costs}
	Fix integers $t\ge4$ and $K\ge t+1$. Let
	\begin{equation}\label{eq:lower-parameters}
		g=1+1/t,\qquad q=g^{t-3},\qquad \chi=g^{t-1}.
	\end{equation}
	We construct an instance $I_{t,K}$ with $2^K-1$ agents and an equal number of chores such that
	\[
	q\frac{K}{K+t-1}\le\gamma(I_{t,K})\le\chi.
	\]
	
	\paragraph{Agents and Weights}
	There are $K$ groups $G_1,\ldots,G_K$, with $n_i=2^{i-1}$ agents each ($i\in\{1,\ldots,K\}$). We refer to $G_1$, which consists of a single agent, as the {\em root group}, and its unique agent as the {\em root agent}. The remaining groups $G_2,\ldots,G_K$ are called {\em nonroot groups}, and their members are referred to as {\em nonroot agents}. This terminology distinguishes the first group from the others, as we will see later on, the root agent will value every chore at its assigned price, whereas a nonroot agent may value chores from her own or earlier blocks differently.
	
	Every agent in $G_i$ has the same unnormalized individual weight
	\begin{equation}\label{eq:lower-weights}
		w_1=t,\qquad w_i=\frac{g^{i-2}}{n_i},\qquad(i\ge2).
	\end{equation}
	Let $W_i=n_iw_i$. Then
	\begin{equation}\label{eq:lower-prefix}
		W_1=t,\quad W_i=g^{i-2}\ (i\ge2),\quad
		R_{j-1}:=\sum_{a<j}W_a=tg^{j-2},\quad
		\sum_{a<j}n_a=n_j-1.
	\end{equation}
	The last two identities hold for $j\ge2$. Individual weights decrease strictly: $w_2=1/2<t$ and $w_{i+1}/w_i=g/2<1$. Total weight is $tg^{K-1}$; a common normalization does not change WMMS. 
	
	\medskip
	\paragraph{Chores and Costs}

	Partition the chores into blocks $M_1,\ldots,M_K$, with $|M_j|=n_j$. For an item $h\in M_j$, every agent in group $G_i$ has cost
	\begin{equation}\label{eq:lower-costs}
		v_i(h)=\frac1{\max\{n_i,n_j\}}
		=\begin{cases}1/n_i,&j\le i,\\1/n_j,&j>i.\end{cases}
	\end{equation}
	The costs are additive, strictly positive, and nonincreasing in the same block order.

	\begin{table}[t]
		\centering
		\begin{tabular}{@{}ccccc@{}}
			\toprule
			Recipient group & $M_1$ & $M_2$ & $M_3$ & $M_4$\\
			\midrule
			$G_1$ & $1$ & $1/2$ & $1/4$ & $1/8$\\
			$G_2$ & $1/2$ & $1/2$ & $1/4$ & $1/8$\\
			$G_3$ & $1/4$ & $1/4$ & $1/4$ & $1/8$\\
			$G_4$ & $1/8$ & $1/8$ & $1/8$ & $1/8$\\
			\bottomrule
		\end{tabular}
		\caption{Per-item costs in the first four blocks of the lower-bound family. Block $M_j$ contains $2^{j-1}$ items. The plateau in row $i$ covers its own and all earlier blocks.}
		\label{tab:lower-costs}
	\end{table}

	\subsection{Upper bounds on WMMS}
	To derive a lower bound on the approximation factor, we need
	\emph{upper bounds} on the WMMS benchmarks: the smaller the
	benchmarks, the less cost any successful allocation is allowed to give to an agent. Before doing that, we need to present two elementary packing lemmas.
	\begin{lemma}[Geometric capacity inequality]\label{lem:geometric}
		For every integer $\ell\ge0$,
		\begin{equation}\label{eq:geometric}
			tg^\ell\ge q(\ell+3).
		\end{equation}
		In particular, $q\le t/3$.
	\end{lemma}
	\begin{proof}
		Set $f_\ell=tg^\ell/(\ell+3)$. Then
		\[
		\frac{f_{\ell+1}}{f_\ell}
		=\frac{(t+1)(\ell+3)}{t(\ell+4)}.
		\]
		The numerator minus the denominator is $\ell+3-t$. The sequence decreases up to index $t-3$, has equal values at $t-3$ and $t-2$, and then increases. Its minimum is $f_{t-3}=g^{t-3}=q$. The last assertion is the case $\ell=0$.
	\end{proof}
	
	If a collection of chores has total cost below the sum of bin
	capacities, integral packing need not follow: each bin might have
	unused space too small for the next whole chore.  The following
	standard reserve compensates for this loss.
	\begin{lemma}[Integral packing of one equal-sized batch]\label{lem:batch}
		Suppose $b$ bins have total capacity $C$ and already assigned load $L$, with no bin overloaded. A batch of $m$ items, each of size $\eta>0$, can be packed integrally into these bins whenever
		\begin{equation}\label{eq:batch}
			C\ge L+m\eta+b\eta.
		\end{equation}
	\end{lemma}
	\begin{proof}
		Place each new item in any bin with residual capacity at least $\eta$. If placement fails after $s<m$ items, every bin has residual capacity less than $\eta$, hence their total residual capacity is less than $b\eta$. But \cref{eq:batch} gives residual capacity at least $(m-s+b)\eta>b\eta$, a contradiction.
	\end{proof}
	
	We are now ready to prove upper bounds on WMMS of agents.
	\begin{lemma}[WMMS bounds]\label{lem:lower-benchmarks}
		Consider the constructed instance $I_{t,K}$. Denote by $\mu_i$ the WMMS of agents in group $G_i$. It holds that
		\begin{equation}\label{eq:lower-benchmarks}
			\mu_1\le\frac tq,\qquad \text{and} \qquad 
			\mu_i\le\frac1{qn_i}\quad(2\le i\le K).
		\end{equation}
	\end{lemma}
	\begin{proof}
		We construct separate whole-item packings for each valuation row into the actual agent bins, whose capacities are the weights in \cref{eq:lower-weights}.
		
		Fix group $i\ge2$ and scale its costs by $a_i=qW_i$. All $2n_i-1$ items in the first $i$ blocks have scaled size $a_i/n_i$. Pack this initial batch using only groups strictly before $i$. Then process successive blocks $M_j$ using only groups strictly before $j$, while $j<i+t-3$ and $j\le K$.
		
		At stage $j=i+\ell$, the cumulative scaled load is
		\[
		C_j=a_i\left(\ell+2-\frac1{n_i}\right).
		\]
		There are $n_j-1$ available bins, and the current item size is $a_i/n_j$; for $j=i$, the current batch is the entire initial prefix. By \cref{eq:lower-prefix,eq:geometric},
		\begin{align*}
			R_{j-1}
			=tW_ig^\ell
			\ge a_i(\ell+3)
			>a_i\left(\ell+2-\frac1{n_i}\right)
			+(n_j-1)\frac{a_i}{n_j}
			=C_j+(n_j-1)\frac{a_i}{n_j}.
		\end{align*}
		All earlier assignments lie in the currently available bins. \Cref{lem:batch} therefore completes each stage integrally, by induction.
		
		For every remaining block $j\ge i+t-3$, we have $W_j=W_i g^{j-i}\ge qW_i=a_i$. Give one item of $M_j$ to each of the $n_j$ bins in $G_j$; it has scaled size $a_i/n_j\le w_j$. Those bins are fresh, since earlier stages used only groups preceding their current block. If the cutoff exceeds $K$, the greedy stages already complete the whole instance. Thus $a_iv_i$ packs inside the weights, and
		\[
		\mu_i\le\frac{w_i}{a_i}=\frac1{qn_i}.
		\]
		
		For the root group $1$, scale its cost by $q$. Put the first item, of scaled size $q\le t/3$, in the sole root bin of capacity $t$. Process blocks $M_j$ for $2\le j\le t-2$ using groups strictly before $j$. At stage $j$, cumulative scaled load is $qj$, item size is $q/n_j$, and there are $n_j-1$ available bins. By \cref{eq:geometric} with $\ell=j-2$,
		\[
		R_{j-1}=tg^{j-2}\ge q(j+1)
		>qj+(n_j-1)q/n_j.
		\]
		Apply \cref{lem:batch}. For $j\ge t-1$, $W_j=g^{j-2}\ge q$, so block $M_j$ fits one item per bin in its own fresh group. This packs $qv_1$ within the weights and gives $\mu_1\le t/q$.
	\end{proof}
	

	\subsection{Proof of Theorem~\ref{thm:lower-main}}
	We establish both lower and upper bounds on the optimal WMMS approximation factor $\gamma(I_{t,K})$ for the instance $I_{t,K}$.
	
	\paragraph{Lower bound}
	Give each item of block $M_j$ price $p_h=1/n_j$. Total price is $K$. Consider an arbitrary allocation and let $\alpha$ be its maximum WMMS ratio. If $\alpha\ge q$, the desired lower bound is immediate. Suppose $\alpha<q$.
	
	For every nonroot group $i\ge2$, an item from its own or an earlier block costs $1/n_i$, whereas
	\[
	\alpha\mu_i\le\frac{\alpha}{qn_i}<\frac1{n_i}.
	\]
	Such an assignment is impossible. Hence every item received by a nonroot agent comes from a later block and has recipient cost equal to its price. The root also values every item at its price. Therefore the allocation's total recipient cost is exactly $K$. By \cref{lem:lower-benchmarks},
	\[
	K\le\alpha\sum_{i=1}^K n_i\mu_i
	\le\frac{\alpha}{q}(t+K-1).
	\]
	It follows that
	\begin{equation}\label{eq:lower-obstruction}
		\gamma(I_{t,K})\ge q\frac{K}{K+t-1}.
	\end{equation}
	The derivation does not restrict how individual chores are distributed within or across blocks. Every forbidden edge was ruled out using the actual WMMS requirement.
	
	\paragraph{Upper bound} We first derive lower bounds on the WMMS of the agents who receive chores in the allocation constructed to prove the lower bound on $\gamma(I_{t,K})$.
	
	\begin{lemma}
		\label{lem:lower-shift}
		For $\chi=g^{t-1}$,
		\[
		\mu_1\ge t/\chi,\qquad
		\mu_i\ge1/(\chi n_i)\quad(2\le i\le K-t+1).
		\]
	\end{lemma}
	\begin{proof}
		For root group, suppose its makespan $\tau<1/\chi$ and consider the first $j=t+1$ blocks. Every one of these items costs at least $1/n_j$. Every bin in group $j$ or later has capacity at most $\tau w_j=\tau\chi/n_j<1/n_j$, since individual weights decrease. Thus all selected items must use earlier bins, of total capacity $\tau R_{j-1}=\tau t\chi<t$. Their total cost is $t+1$, a contradiction. Hence $\mu_1=t\tau\ge t/\chi$.
		
		For nonroot group $i\ge2$, put $j=i+t-1\le K$. If $\tau<1/(\chi W_i)$, no item in the first $j$ blocks fits in group $j$ or later, because each costs at least $1/n_j$ while $\tau w_j=\tau\chi W_i/n_j<1/n_j$. Earlier capacity is $\tau R_{j-1}=\tau t\chi W_i<t$. But the selected row-$i$ cost is $2-1/n_i+j-i=t+1-1/n_i>t$. Thus $\tau\ge1/(\chi W_i)$, and $\mu_i=w_i\tau\ge1/(\chi n_i)$.
	\end{proof}
	
	Give the first $t$ blocks to the root. Its load is $t\le\chi\mu_1$. For each $2\le i\le K-t+1$, give block $M_{i+t-1}$ to group $G_i$, with exactly $2^{t-1}$ items per agent. This is integral because $n_{i+t-1}=2^{t-1}n_i$. Each recipient's cost is $1/n_i\le\chi\mu_i$. The last $t-1$ groups receive nothing. All chores are allocated, proving 
	\[
	\gamma(I_{t,K})\le\chi=g^{t-1}.
	\]
	Together with \cref{eq:lower-obstruction}, this proves \cref{eq:main-sandwich} and thus Theorem~\ref{thm:lower-main}.

	\section{Conclusion}\label{sec:conclusion}
	We studied the allocation of indivisible chores among agents with arbitrary positive weights and additive nonnegative costs, with the goal of guaranteeing each agent a bounded multiple of her weighted maximin share. We showed that every instance admits an $\alpha_\star$-WMMS allocation with $\alpha_\star<4.124234528$, and that a $4.12424$-WMMS allocation can be computed in polynomial time for rational inputs. The upper bound requires neither rounding of weights or costs nor a common ordering of the agents' cost functions; it relies only on the existence of a row-wise integral packing for each agent. On the lower-bound side, we constructed instances with strictly positive, identically ordered costs whose optimal WMMS factors converge to $e$. Together, these results narrow the previously best-known gap of $[2,12]$ to $[e,\,4.124234528]$.
	
	Closing the remaining gap is the main open problem, and our analysis indicates where progress might come from. On the upper-bound side, the guarantee $\alpha(\rho,r)=\rho+\beta(\rho,r)$ consists of two separate losses: the fractional bound $\beta(\rho,r)\approx 3.0218$ obtained from the multiscale dual-price inequalities, and an additive loss of $\rho\approx 1.1024$ incurred by rounding, which charges one eligible chore per agent. The two parts are optimized jointly through the eligibility threshold $\rho$, but they are analyzed independently. An improvement could come from stronger packing inequalities than the capped-price inequalities of \cref{lem:capped}, from a rounding procedure that exploits the structure of the WMMS packings rather than only the fractional loads and the largest supported costs, or from treating large and small chores jointly. We leave open whether either loss is unavoidable for approaches of this type.
	
	On the lower-bound side, our obstruction combines two ingredients: individually infeasible agent--chore pairs and a linear comparison between the total price of the chores and the sum of the permitted budgets. As shown in \cref{sec:lower}, the constructed family itself admits allocations whose factors converge to $e$, so this construction cannot yield a stronger bound. Improving on $e$ will therefore require instances with a different structure, and possibly an argument that captures integrality obstructions beyond price-versus-budget counting. It also remains open whether the identically ordered case, to which our lower-bound instances belong, is as hard as the general case.
	
	\paragraph{Use of generative AI.}
	The authors used GenAI iteratively to refine and improve an initial draft of this paper. The authors have verified all proofs and accept full responsibility for the accuracy and correctness of the results presented in the paper.

	\bibliographystyle{plain}
	\bibliography{ref}

	\newpage
	
	\section{Appendix}
	
	\begin{proof}[Proof of Lemma~\ref{lem:compute:alpha}]
		Set $x=r/\rho>1$ and $\ell=1+\ln\rho$.  Then
		\begin{equation}\label{eq:upper-alpha-x}
			\alpha(\rho,\rho x)
			=\rho\left(1+\frac{x(1+\ln x)}{\ell x-1}\right).
		\end{equation}
		For fixed $1<\rho<e$, the derivative of the fractional term in $x$
		has the sign of
		\[
		x(1+\ln\rho)-2-\ln x.
		\]
		It vanishes at a unique $x>1$.  At that value,
		\[
		x(1+\ln\rho)=2+\ln x,
		\qquad \beta(\rho,\rho x)=\rho x.
		\]
		Express $\rho=\exp((2+\ln x)/x-1)$ and minimize the remaining
		one-variable expression
		\[
		(1+x)\exp\left(\frac{2+\ln x}{x}-1\right).
		\]
		Its unique interior minimizer satisfies
		\begin{equation}\label{eq:upper-stationary-condition}
			x^2=(1+x)(1+\ln x).
		\end{equation}
		This gives the numerical parameter values
		\[
		\rho_*\approx1.1024340960285306,
		\quad r_*\approx3.0218004310998404,
		\quad\alpha_*\approx4.124234527128371.
		\]
		The values of $\rho_*$ and $r_*$ are \emph{not} required as exact
		inputs in the existence proof.
	\end{proof}
	
\end{document}